%% file: size-dep-arxiv.tex
\documentclass[12pt]{article}

\usepackage{hyperref}
\usepackage{amsmath,amssymb,amsfonts,amsthm}
\usepackage{graphicx}

\newtheorem{theorem}{Theorem}[section]

\newcommand{\prodasm}[1]{\mathcal{A}[\mathcal{#1}]}

\newcommand{\dom}[1]{{\rm dom}(#1)}
\newcommand{\ta}{\tilde{\alpha}}
\newcommand{\tb}{\tilde{\beta}}
\newcommand{\tg}{\tilde{\gamma}}
\newcommand{\str}[1]{{\rm str}(#1)}
\newcommand{\NP}{{\ensuremath{\mathsf{NP}}}}
\newcommand{\coNP}{{\ensuremath{\mathsf{coNP}}}}
\newcommand{\PSPACE}{{\ensuremath{\mathsf{PSPACE}}}}

\begin{document}

\title{Size-Dependent Tile Self-Assembly:\\Constant-Height Rectangles and Stability}
\author{S\'andor P. Fekete\thanks{TU Braunschweig, \protect{\url{s.fekete@tu-bs.de}}}
\and Robert T. Schweller\thanks{University of Texas--Pan American, \protect{\url{rtschweller@utpa.edu}}}
\and Andrew Winslow\thanks{Universit\'e libre de Bruxelles, \protect{\url{awinslow@ulb.ac.be}}}}
\date{}

\maketitle

\begin{abstract}
We introduce a new model of algorithmic tile self-assembly called \emph{size-dependent assembly}.
In previous models, supertiles are stable when the total strength of the bonds between any two halves exceeds some constant temperature.
In this model, this constant temperature requirement is replaced by an nondecreasing \emph{temperature function} $\tau : \mathbb{N} \rightarrow \mathbb{N}$ 
that depends on the size of the smaller of the two halves.
This generalization allows supertiles to become unstable and break apart, 
and captures the increased forces that large structures may place on the bonds holding them together.

We demonstrate the power of this model in two ways.
First, we give fixed tile sets that assemble constant-height rectangles and squares of arbitrary input size given an appropriate temperature function. 
Second, we prove that deciding whether a supertile is stable is \coNP-complete. 
Both results contrast with known results for fixed temperature. 
\end{abstract}

\input{introduction}

\input{definitions}

\input{line}

\input{square}

\input{instability}

\input{openProblems}

\section*{Acknowledgements}

This work began at the Bellairs Workshop on Self-Assembly and Computational Geometry, March 21-28th, 2014.
We thank the other participants for a productive and positive atmosphere,
in particular, Alexandra Keenan and the co-organizers, Erik Demaine and Godfried Toussaint.

\bibliographystyle{abbrv}
\bibliography{size-dep-arxiv}

\end{document}

%% file: introduction.tex
\section{Introduction}

In this paper, we introduce the \emph{size-dependent tile self-assembly model}, a natural extension of the well-studied \emph{two-handed tile assembly model} or \emph{2HAM}~\cite{Cannon-2013a}.
As in the 2HAM, a size-dependent system consists of a collection of square Wang tiles~\cite{Robinson-1971,Wang-1961} with an associated \emph{bond strength} assigned to each tile edge color.
In the 2HAM, self-assembly proceeds by repeatedly combining any two previously assembled \emph{supertiles} into a new \emph{stable} supertile provided the total bond strength between the supertiles meets or exceeds some positive integer called the \emph{temperature}.

Although the 2HAM is both simple and natural, the model does not capture the intuition that two large assemblies should require more bond strength to be stable than two very small assemblies.
As an analogy, a single staple is sufficient to attach two pieces of paper or to attach a sheet of paper to the hull of a battleship.
However, a staple is too weak to amalgamate together two battleships.

The size-dependent self-assembly model generalizes the 2HAM by replacing the fixed, integer temperature parameter $\tau$ of the 2HAM with a nondecreasing \emph{temperature function} $\tau(n)$ that specifies a required threshold of bond strength when given the size of the smaller of two supertiles under consideration.
A set of tile types and temperature function together define a size-dependent self-assembly system.

\paragraph{Our results.}
We first consider efficiently assembling fixed-height rectangles and squares in the size-dependent self-assembly model.
We prove that there exists a fixed tile set assembling a $k \times 3$ rectangle for every $k \geq 7$ given an appropriate temperature function.
This tile set is extended to obtain a matching result for $k \times k$ squares.
These results demonstrate that size-dependent temperature functions can, in theory, direct assembly in the spirit of temperature programming~\cite{Kao-2006a,Summers-2009}, concentration programming~\cite{Becker-2006,Doty-2010c,Kao-2008a}, and staging~\cite{Demaine-2008b}.
Unlike these other methods, size-dependence is present in all physical systems, but has not be demonstrated to be programmable.
Thus these constructions demonstrate that this ubiquitous aspect of physical systems can (and likely already does) direct assembly in dramatic ways, regardless of whether they can be implemented physically.

In addition to the design of systems that assemble rectangles and squares, we consider the complexity of determining if a supertile is \emph{stable}, i.e. cannot break apart due to insufficient bond strength.
Determining the stability of an supertile is a fundamental problem for design, simulation, and analysis of tile self-assembly systems.
This problem enjoys a straightforward, polynomial-time solution in the 2HAM.
In contrast, we prove that the problem is \coNP-complete in the size-dependent model, even for temperature functions with just two distinct temperatures. 

\paragraph{Reversibility.}
A key feature of size-dependence is \emph{reversibility}: the possibility of breaking bonds.
Our rectangle and square constructions make critical use of reversibility to beat tile type lower bounds in similar models (see~\cite{Rothemund-2000a}), and our hardness result proves that this mechanism is capable of complex behaviors.

Reversibility has been more directly incorporated into a number of other self-assembly models via glues that repel~\cite{Doty-2010d,Reif-2006} or deactivate~\cite{Hendricks-2013a,Keenan-2013,Padilla-2013a}, tiles that dissolve~\cite{Abel-2010a}, and temperatures that change over time~\cite{Aggarwal-2005a,Summers-2009}.
Reversibility in these models has yielded a number of new functionalities, including replication~\cite{Abel-2010a,Keenan-2013}, fuel-efficient computation~\cite{Padilla-2013a,Schweller-2013a}, shape identification~\cite{Patitz-2012b}, and efficient small-scale assembly of general shapes~\cite{Demaine-2011b}.
We believe that further study of the ubiquitous but indirect form of reversibility found in size-dependent self-assembly may yield similar functionality.

%% file: definitions.tex
\section{Definitions}
\label{sec:definitions}

The first three subsections define the 2HAM, giving definitions equivalent to those in prior work, e.g.~\cite{Cannon-2013a}. 
The final section describes the differences between the two-handed and size-dependent models.

\subsection{Tiles, assemblies, and supertiles}

A \emph{tile type} is a quadruple $(g_N, g_E, g_S, g_W)$ of \emph{glues} from a fixed alphabet $\Sigma$.
Each glue $g_i \in \Sigma$ has an associated non-negative integer \emph{strength}, denoted by $\str{g_i}$.\footnote{In later sections, glues with strength~0 are treated as non-existent.}
An instance of a tile type, called a \emph{tile}, is an axis-aligned unit square with center in $\mathbb{Z}^2$.
The edges of a tile are labeled with the glues of the tile's type (e.g. $g_N$, $g_E$, $g_S$, $g_W$) in clockwise order, starting with the edge with normal vector $\langle 0, 1 \rangle$.
Two tiles are \emph{adjacent} if their centers have distance~1.

An \emph{assembly} $\alpha$ is a partial mapping $\alpha : \mathbb{Z}^2 \rightarrow T$ from tile locations to a set of tile types $T$, also called a \emph{tile set}.
The domain of this partial function is denoted by $\dom{\alpha}$.
Each assembly has a dual \emph{bond graph}: a grid graph with vertex set $\dom{\alpha}$ and an edge between every pair of adjacent tiles that form a bond.
An edge cut of the bond graph of an assembly is also called a \emph{cut} of the assembly, and the total strength of the bonds of the edges in the cut is the \emph{strength} of the cut.
An assembly is \emph{$\tau$-stable} if every cut of the assembly has strength at least $\tau$.

For an assembly $\alpha: \mathbb{Z}^2 \rightarrow T$ and vector $\vec{u} = \langle x, y \rangle$ with $x, y \in \mathbb{Z}^2$, the assembly $\alpha + \vec{u}$ denotes 
the assembly consisting of the tiles in $\alpha$, each translated by $\vec{u}$.
For two assemblies $\alpha$ and $\beta$, $\beta$ is a \emph{translation} of $\alpha$, written $\beta \simeq \alpha$, provided that there exists a vector $\vec{u}$ such that $\beta = \alpha + \vec{u}$.
The \emph{supertile} of $\alpha$ is the set $\ta = \{ \beta : \alpha \simeq \beta \}$.
A supertile $\ta$ is \emph{$\tau$-stable} provided that the assemblies it contains are $\tau$-stable.
The \emph{size} of a supertile is denoted by $|\ta|$ and is equal to the size of an assembly in $\ta$ (and not the cardinality of $\ta$, which is always $\aleph_0$).

\subsection{The assembly process}

Two assemblies $\alpha$ and $\beta$ are \emph{disjoint} if $\dom{\alpha} \cap \dom{\beta} = \varnothing$.
The \emph{union} of two disjoint assemblies $\alpha$ and $\beta$, denoted by $\alpha \cup \beta$, is the partial function $\alpha \cup \beta : \mathbb{Z}^2 \rightarrow T$ defined as $(\alpha \cup \beta)(x, y) = \alpha(x, y)$ if $(x, y) \in \dom{\alpha}$ and $(\alpha \cup \beta)(x, y) = \beta(x, y)$ if $(x, y) \in \dom{\beta}$.
Two supertiles $\ta$ and $\tb$ can \emph{combine} into a supertile $\tg$ provided:
\begin{itemize}
\item There exist disjoint assemblies $\alpha \in \ta$ and $\beta \in \tb$.
\item $\alpha \cup \beta = \gamma \in \tg$ and the cut partioning $\dom{\gamma}$ into $\dom{\alpha}$ and $\dom{\beta}$ has strength at least $\tau$ (equivalently, $\gamma$ is $\tau$-stable).
\end{itemize}
The set of all combinations of $\ta$ and $\tb$ at temperature $\tau$ is denoted by $C^\tau_{\ta, \tb}$.

\subsection{Two-handed tile assembly systems}

A \emph{two-handed tile assembly system} or \emph{two-handed system} is a pair $\mathcal{T} = (T,\tau)$, where $T$ is a \emph{tile set} and $\tau \in \mathbb{N}$ is a \emph{temperature}.
Given a system $\mathcal{T}=(T,\tau)$, a supertile $\ta$ is \emph{producible}, written $\ta \in \prodasm{T}$, provided that either $|\ta| = 1$ or $\ta$ is a combination of two other producible supertiles of $\mathcal{T}$.
A supertile $\ta$ is \emph{terminal} provided that for all producible supertiles $\tb$, $C^\tau_{\ta,\tb} = \varnothing$.
A system is \emph{directed} or \emph{deterministic} provided that it has only one terminal supertile.

Given a \emph{shape} $P \subseteq \mathbb{Z}^2$, we say a system $\mathcal{T}$ \emph{self-assembles $P$}, provided that every terminal supertile $\ta$ of $\mathcal{T}$ has an assembly $\alpha \in \ta$ such that $\dom{\alpha} = P$.
That is, every terminal supertile has shape $P$, up to translation.
A shape $P$ is a \emph{$w \times h$ rectangle} provided that $P = \{x+1, x+2, \dots, x+w\} \times \{y+1, y+2, \dots, y+h\}$ for some $x, y, w, h \in \mathbb{Z}$.
If $w = h$, then the rectangle is a \emph{square}.

\subsection{Size-dependent systems}

A \emph{size-dependent two-handed tile assembly system} or \emph{size-dependent system} $\mathcal{S} = (T, \tau)$ is a generalization of a two-handed tile assembly system.
Two-handed and size-dependent systems are identical, except for the definition of $\tau$.
Recall that in two-handed systems, $\tau \in \mathbb{N}$ determines the bond strength needed for two supertiles to combine and for a supertile to be $\tau$-stable.

In size-dependent systems, $\tau$ is not an integer temperature, but rather a nondecreasing \emph{temperature function} $\tau : \mathbb{N} \rightarrow \mathbb{N}$.
An assembly $\gamma$ is \emph{$\tau$-stable} provided any cut partioning $\dom{\gamma}$ into two assemblies $\dom{\alpha}$, $\dom{\beta}$ has strength at least $\tau(\rm{min}(|\alpha|, |\beta|))$.
A supertile $\tg$ is \emph{$\tau$-stable} provided the assemblies in $\tg$ are $\tau$-stable.  
Also, two supertiles $\ta$ and $\tb$ can \emph{combine} into a supertile $\tg$ provided that:
\begin{itemize}
\item There exist disjoint assemblies $\alpha \in \ta$ and $\beta \in \tb$.
\item $\alpha \cup \beta = \gamma \in \tg$ and the cut partioning $\dom{\gamma}$ into $\dom{\alpha}$ and $\dom{\beta}$ has strength at least $\tau(\rm{min}(|\alpha|, |\beta|))$.
\end{itemize}
For a given temperature function $\tau : \mathbb{N} \rightarrow \mathbb{N}$, the set of all combinations of $\ta$ and $\tb$ is denoted by $C^\tau_{\ta, \tb}$.
Note that the second condition is not equivalent to $\gamma$ being $\tau$-stable.
Figure~\ref{fig:become-unstable} illustrates an example: a cut in a supertile has sufficient strength, but combining with another supertile causes increased size that causes the cut to become insufficiently strong.
So $\ta$, $\tb$ may be $\tau$-stable while their combination $\tg$ is \emph{$\tau$-unstable}.

\begin{figure}[ht]
\centering
\includegraphics[scale=1.0]{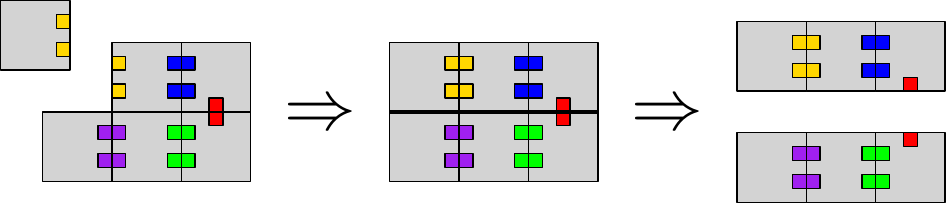}
\caption{Three steps of size-dependent self-assembly with glue function $\tau(n) = n-1$.
The addition of a new tile (left) causes the supertile to have a strength-1 cut partioning it into two supertiles of~3 tiles each (center).
Because $\tau(3) = 2 > 1$, the supertile can then break (right).}
\label{fig:become-unstable}
\end{figure}

Supertiles that are $\tau$-unstable can also ``break'' into smaller supertiles.  
A supertile $\tg$ can \emph{break} into $\ta$ and $\tb$ provided that:
\begin{itemize}
\item There exist disjoint assemblies $\alpha \in \ta$ and $\beta \in \tb$ with connected bond graphs.
\item $\alpha \cup \beta = \gamma \in \tg$ and the strength of the cut partioning $\gamma$ into $\alpha$ and $\beta$ is less than $\tau(\rm{min}(|\alpha|, |\beta|))$.
\end{itemize}
A cut between two supertiles resulting from a break is called a \emph{break cut}. 
For a given temperature function $\tau : \mathbb{N} \rightarrow \mathbb{N}$, the set of all supertiles resulting from breaks of $\tg$ is denoted by $B^{\tau}_{\tg}$. 
Given a size-dependent system $\mathcal{T} = (T, \tau)$, a supertile $\ta$ is \emph{producible} provided either:
\begin{itemize}
\item $|\ta| = 1$.
\item $\ta$ is the combination of two other producible supertiles.
\item $\ta$ is the result of a break of a producible supertile.
\end{itemize}
A producible supertile $\ta$ is \emph{terminal} provided $C^{\tau}_{\ta, \tb} = \varnothing$ and $B^{\tau}_{\ta} = \varnothing$.

Note that the conditions on supertiles combining and breaking do \emph{not} imply that combining supertiles or supertiles resulting from a break are $\tau$-stable.
This allows for systems with an infinite number of producible supertiles and a unique terminal supertile, including those described in this work.


%% file: line.tex
\section{Constant-Height Rectangles}
\label{sec:line}

Here we prove that there exists a single set of tiles that can be used to self-assemble constant-height rectangles of arbitrary width using an appropriate choice of temperature function.
Such a result contrasts with the polynomial number of tiles required to assemble a constant-height rectangle in an assembly system with constant temperature~\cite{Aggarwal-2005a}.

\begin{theorem}
\label{thm:lazy-line}
There exists a tile set $T$ such that for every $k \geq 7$, there exists a size-dependent system with tile set $T$ that self-assembles a $k \times 3$ rectangle.
\end{theorem}

\begin{proof}
The temperature function used is:
\begin{displaymath}
\tau(n) = \left\{
\begin{array}{ll}
3 & : n \leq k-6\\
4 & : k-5 \leq n \leq k+3\\
5 & : k+4 \leq n \leq 2k-2\\
8 & : {\rm otherwise}
\end{array}
\right.
\end{displaymath}

The tile set consists of three tile types and two \emph{blocks}: supertiles with unique internal glues and strength~8, the maximum temperature of the system.
The tiles and blocks are listed and named in Figure~\ref{fig:lazy-line-blocks}.

\begin{figure}[ht]
\centering
\includegraphics[scale=1.0]{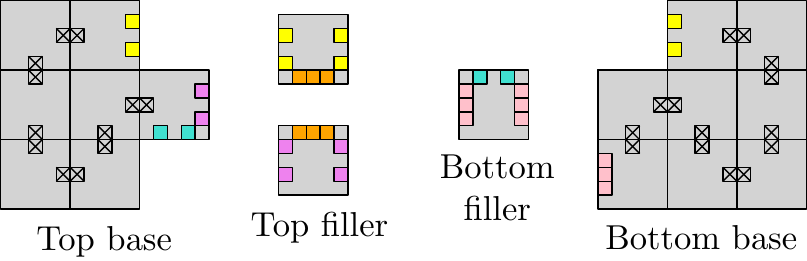}
\caption{The tile types and blocks for the constant-height rectangle construction.
The gray glues are unique and strength at least 8.}
\label{fig:lazy-line-blocks}
\end{figure}

The system works by assembling a unique terminal $k \times 3$ supertile in three phases.
First, top filler tiles and top bases combine into arbitrarily wide height-2 supertiles.
These undergo at least two breaks to form \emph{top half} supertiles of size $2k-3$.
Second and separately, bottom filler tiles and bottom bases combine to form \emph{bottom half} supertiles of size approximately $k+3$.
Finally, these two halves combine into a terminal $k \times 3$ supertiles shown in Figure~\ref{fig:lazy-line-final}.
It can easily be verified that this supertile is a terminal supertile of the system; it remains to be shown that no other terminal supertiles of the system exist (necessary for the system to \emph{self-assemble} a $k \times 3$ rectangle).

\begin{figure}[ht]
\centering
\includegraphics[scale=1.0]{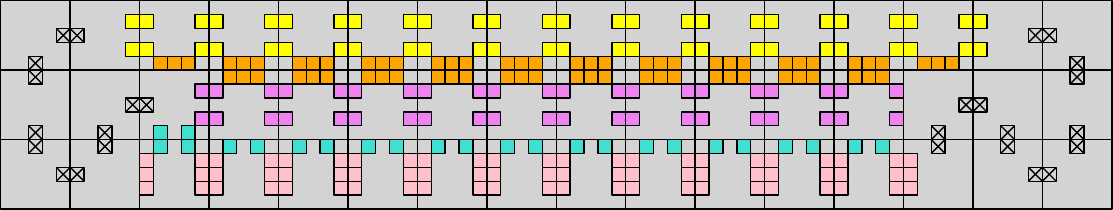}
\caption{The unique terminal supertile of the constant-height rectangle construction.}
\label{fig:lazy-line-final}
\end{figure}

\textbf{Top filler supertiles.}
To start, consider the producible supertiles consisting of only top filler tiles, called \emph{top filler} supertiles.
Because $\tau(n) > 2$ for all $n$, upper and lower top filler tiles must first combine into size-2 supertiles before combining with other top filler supertiles into height-2 rectangular supertiles (lower right supertile in Figure~\ref{fig:lazy-line-filler}).
These rectangular supertiles break along 2-edge and 3-edge cuts into the remaining supertiles seen in Figure~\ref{fig:lazy-line-filler}.

Because $k \geq 7$, any partition of the lower right supertile in Figure~\ref{fig:lazy-line-filler} either has a part that is a single tile or uses a strength-4 cut of at least~2 edges and thus both parts have size at least $k+3 \geq 10$.
Therefore, the remaining~8 types of supertiles in Figure~\ref{fig:lazy-line-filler} have at least~4 columns of~2 tiles each.

\begin{figure}[ht]
\centering
\includegraphics[scale=1.0]{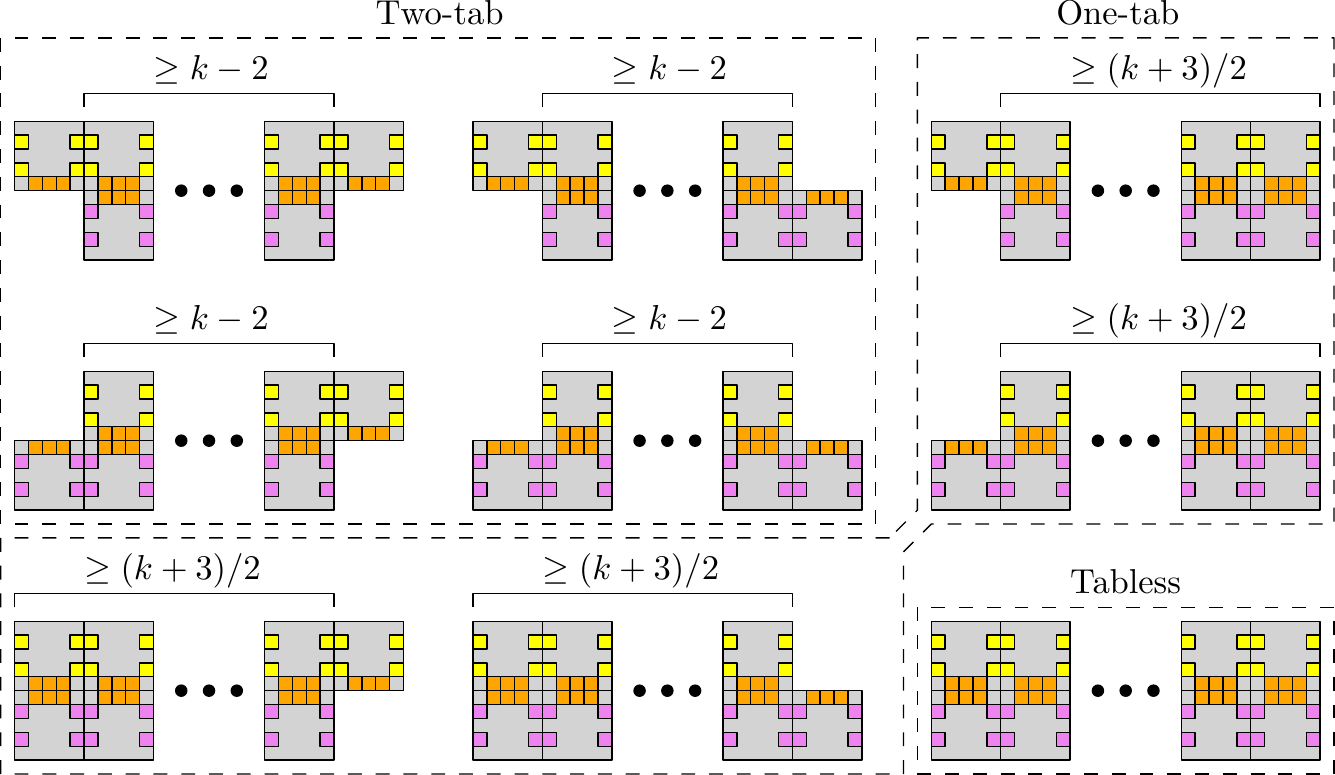}
\caption{The producible top filler supertiles.}
\label{fig:lazy-line-filler}
\end{figure}

The width bounds seen in the figure are computed by considering how the supertiles are created.
If the supertile is the result of a break, it must satisfy the size bound for the strength of the cut used in the break.
If it is the result of a combination, it must be larger than the total sizes of the combined supertiles.\footnote{An upper bound is also implied by $\tau$, but this is ignored here.}

We designate three types of top filler supertiles as seen in Figure~\ref{fig:lazy-line-filler}.
As already proven, breaks only result in single tiles or supertiles of size 10 and larger.
Any two-tab (one-tab) supertile can break into a one-tab (tabless) supertile and a single tile, and these are the only breaks that use cuts of strength at most~3.
Then any other break uses a cut of strength~4 or more, and so results in supertiles of size at least $k+4$.
Thus any combination of two-tab and one-tab supertiles has size at least $2(k+4)$.
A two-tab supertile can also be the result of a break using a cut of strength~7 and thus have size at least $2k-3$ and, because two-tab supertiles have even size, $2k-2$.
Because ${\rm min}(2(k+4), 2k-2)-1 = 2k-3$ and $k \geq 7$, $2k-3 \geq k+4$ and a break of a two-tab supertile into a single tile and one-tab supertile cannot yield a one-tab supertile smaller than $k+4$.
In conclusion, one-tab and two-tab supertiles have size at least $k+4$ and $2k-2$, respectively, implying the bounds seen in Figure~\ref{fig:lazy-line-filler}.

\textbf{Top and bottom halves.}
Top filler supertiles cannot combine with other supertiles, except for a complete top base to form a \emph{top half} supertile (upper supertile in Figure~\ref{fig:lazy-line-halves}).
Top half supertiles may combine with top filler supertiles and break into top half and top filler supertiles.
A top half supertile with a single upper filler tile in the rightmost column is \emph{ready}.
Because ready top half supertiles are two-tab top filler supertiles that have combined with a top base, they have size at least $2k-3$ and thus width at least $k-2$.

Independently of top halves, bottom filler tiles combine into arbitrarily wide height-1 supertiles called a \emph{bottom filler} supertile.
These supertiles also combine with bottom bases at various stages of assembly.
A \emph{bottom half} supertile contains bottom filler tiles and a completed bottom base.
If the number of bottom filler tiles in a bottom half is at least $2k-18$ (and there exists a 1-edge strength-3 cut partitioning the supertile into two of size at least $k-5$), the bottom half can break into a bottom half and bottom filler supertile.

\textbf{Combining halves.}
The only shared glues between top and bottom tiles are the strength-2 glues on the south of the top base and west of the bottom base (turquoise and yellow in Figure~\ref{fig:lazy-line-blocks}).
Thus a supertile consisting of bottom tiles cannot combine with a supertile consisting of top tiles, unless the supertiles are bottom and top halves.

\begin{figure}[ht]
\centering
\includegraphics[scale=1.0]{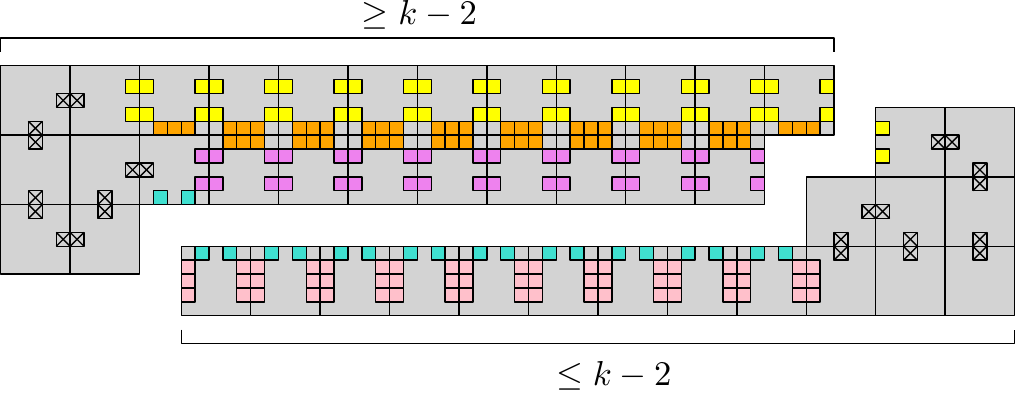}
\caption{The top half and bottom half supertiles. The bottom half $\lambda$ can be arbitrarily large, but the upper bound follows from the requirement that to combine, $\tau(|\lambda|) \leq 4$ and thus $|\lambda| \leq k$.}
\label{fig:lazy-line-halves}
\end{figure}

A bottom half and top half can combine, provided they have the same width and the top half is ready (and thus has width at least $k-2$.
Moreover, because the maximum strength of the bonds between the bottom and top halves is~4, they can only combine only if the smaller supertile, necessarily the bottom half, has size at most $k+3$ and thus width at most $k-2$.
Thus, the bottom and top halves combine provided they both have width exactly $k-2$, forming a terminal supertile of width exactly $k$.

\textbf{No waste.}
Although it is not required by the definition of self-assembly, this system also has the property that every supertile may undergo a sequence of breaks and combinations to become terminal.
In other words, the system has no ``waste'' supertiles.
This can be seen by noting that supertiles not found within the (unique) terminal supertile, i.e. top filler supertiles wider than $k-4$, top halves wider than $k-2$, bottom filler supertiles of width more than $k-5$, and bottom halves of width more than $k-2$ can repeatedly break into smaller supertiles that \emph{are} found in the terminal supertile.
\end{proof}

The temperature functions used in the previous construction have values bounded above by the constant~8.
Next, we prove that any set of temperature functions used to assemble arbitrarily large constant-height rectangles are similarly bounded above by a constant.

\begin{theorem}
\label{thm:line-lb}
Let $T$ be a tile set and $\tau_1, \tau_2, \dots$ be an infinite sequence of temperature functions such that the size-dependent system $(T, \tau_i)$ assembles a $k_i \times O(1)$ rectangle and all $k_i$ are distinct.
Let $f(n) = {\rm min}_{i \in \mathbb{N}}(\tau_i(n))$. 
Then $f(n) = O(1)$.
\end{theorem}

\begin{proof}
Let $c \in \mathbb{N}$ be the maximum height of a rectangle assembled by a system $(T, \tau_i)$.
Let $g_{\rm max}$ be the maximum strength of a glue in $T$.
Let $\tg$ be a terminal assembly of $(T, \tau_i)$ and thus a rectangle with width $k_i$.
For any $n \leq k_i/2$, there exists a cut of $\tg$ into supertiles $\ta$, $\tb$ such that $n = |\ta| \leq |\tb|$ and the cut contains at most $c+1$ edges.
Then since $\tg$ is stable, $f(n) \leq \tau_i(n) \leq (c+1)g_{\rm max}$ for all $n \leq k_i/2$.
Because there exist infinitely many $k_i$, every $n$ has $n \leq k_i/2$ for large enough $k_i$ and we conclude that $f(n) \leq (c+1)g_{\rm max}$ for all $n \in \mathbb{N}$.
\end{proof}

%% file: square.tex
\section{Squares}
\label{sec:squares}

Here we extend the constant-height rectangle construction in the last section to assemble squares.
The temperature function, tile types, and blocks from the constant-height rectangle construction are used to form the base of the square; additional tile types and blocks are used to ``fill in'' the remainder of the square once the base is complete.

\begin{theorem}
\label{thm:lazy-square}
There exists a tile set $T$ such that, for every $k \geq 7$, there exists a size-dependent system with tile set $T$ that self-assembles a $k \times k$ square.
\end{theorem}

\begin{proof}
The temperature function is identical to that of the construction in the proof of Theorem~\ref{thm:lazy-line}.
The function is:
\begin{displaymath}
\tau(n) = \left\{
\begin{array}{ll}
3 & : n \leq k-6\\
4 & : k-5 \leq n \leq k+3\\
5 & : k+4 \leq n \leq 2k-2\\
8 & : {\rm otherwise}
\end{array}
\right.
\end{displaymath}

The tile set consists of the six tile types and three blocks listed and named in Figure~\ref{fig:lazy-lt-square-blocks}.

\begin{figure}[ht]
\centering
\includegraphics[scale=1.0]{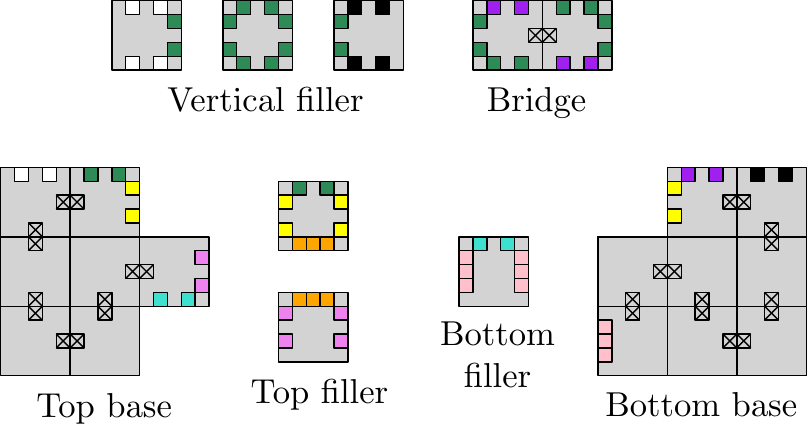}
\caption{The tile types and blocks for the square construction.
The gray glues are unique and have strength~8.}
\label{fig:lazy-lt-square-blocks}
\end{figure}

The system functions by first assembling the bottom three rows of the square, then filling in the remainder using the width of these rows as a scaffold.
The first three rows are assembled identically to the $k \times 3$ rectangle in the proof of Theorem~\ref{thm:lazy-line} -- refer to this proof for details.
Inspecting the glues of the tiles and blocks is sufficient to observe that vertical filler tiles can only combine with a supertile already containing a vertical filler tile or bridge, and a bridge can only combine with a supertile containing both top and bottom tiles.
As established in the proof of Theorem~\ref{thm:lazy-line}, the only such supertile is a unique $k \times 3$ rectangle supertile, called a \emph{scaffold} supertile.

\begin{figure}[ht]
\centering
\includegraphics[scale=1.0]{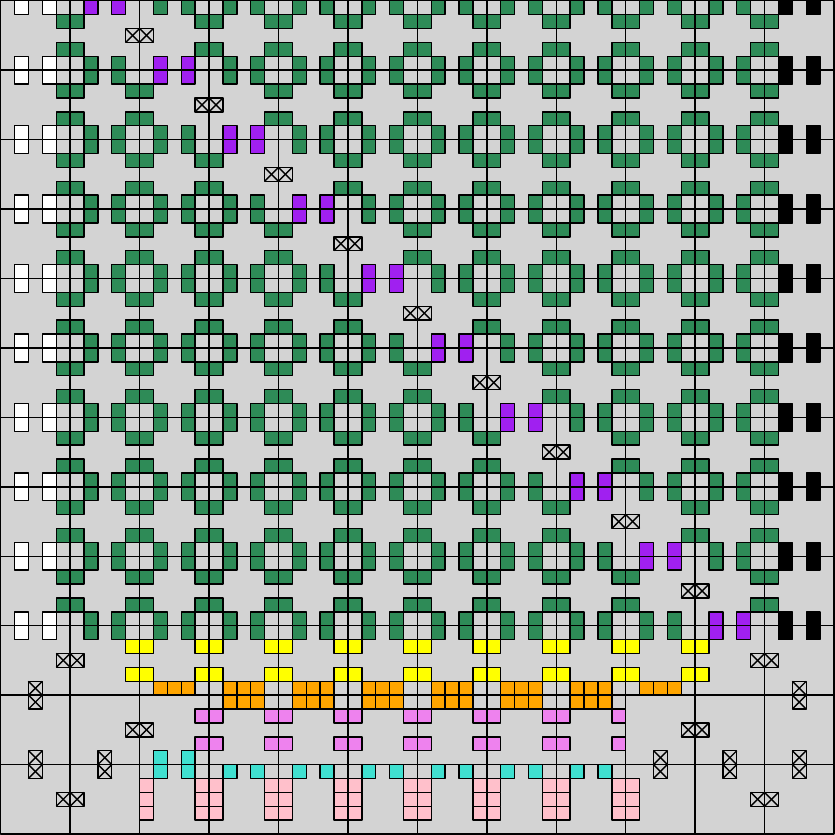}
\caption{The unique terminal supertile of the square construction.}
\label{fig:lazy-lt-square-final}
\end{figure} 

A bridge can combine with a scaffold to initiate a sequence of combinations between a growing supertile that ``fills in'' upwards, called an \emph{accumulator} supertile, and bridge and vertical filler tiles.
The result of this sequence is the terminal $k \times k$ supertile in Figure~\ref{fig:lazy-lt-square-final}.
There are many possible sequences of these combinations and thus many accumulators; we claim every accumulator is $\tau$-stable.
Provided this claim holds, the system has a unique terminal $k \times k$ supertile and thus self-assembles a $k \times k$ rectangle.
The remainder of the proof is dedicated to this claim.
 
Because the vertical filler tiles have exclusively strength-2 glues, any row of an accumulator that contains vertical filler tiles must contain a bridge.
Even more, the set of tiles in any row of the accumulator must be contiguous. 

Notice that $\tau(n) \leq 8$ for all $n$, and every glue has strength at least~2 and every pair of adjacent tiles, excluding those in bases, share a bond.
Consider the restrictions on break cuts of the accumulator.
Since break cuts cannot contain strength-8 bonds, no break cut enters a base.
Thus any break cut is equivalent to a path of at most~3 coincident edges shared by adjacent tiles that starts and ends on the boundary of the accumulator. 
The boundary of the accumulator has four ``sides'': the bottom, consisting of horizontal tile edges, the left and right, consisting of a vertical tile edges, and the top, consisting of both types of edges.
Consider the paths of length at most~3 between sides of the accumulator. 

A path that starts and ends on the same side of the accumulator has length~3 and the resulting partition has a single-tile part.
Then because $\tau(1) = 3 < 3 \cdot 2$, this cannot be a break cut.
A path that starts and ends on adjacent sides of the accumulator has length~2 or~3 and the resulting partition has a part of at most two tiles.
Then because $\tau(1) \leq \tau(2) \leq 4 \leq 2 \cdot 2$, this cannot be a break cut.
Finally, a path that starts and ends on opposite sides of the accumulator has length~3 and must be~3 vertical edges.
Because vertical filler tiles in a row are contiguous, one of the parts must contain no vertical filler tiles and thus is the smaller of the two parts.
This cut cannot be a break cut, as otherwise the $k \times 3$ scaffold supertile is $\tau$-unstable, contradicting a claim in the proof of Theorem~\ref{thm:lazy-line}.
Thus there are no break cuts and all accumulators are $\tau$-stable.
\end{proof}


%
The constant-height rectangle construction used as the basis for the construction of Theorem~\ref{thm:lazy-square} result in temperature functions that are bounded above by a constant.
We conjecture that there exists a square construction that uses temperature functions that all scale as $\Omega(\sqrt{n})$, and prove that no better lower bound is possible:

\begin{theorem}
\label{thm:square-lb}
Let $T$ be a tile set and $\tau_1, \tau_2, \dots$ be an infinite sequence of temperature functions such that the size-dependent system $(T, \tau_i)$ assembles a $k_i \times k_i$ square and $k_i$ are all distinct.
Let $f(n) = {\rm min}_{i \in \mathbb{N}}(\tau_i(n))$, the minimum of all temperature functions for size $n$.
Then $f(n)$ is not $\omega(\sqrt{n})$.
\end{theorem}

\begin{proof}
Let $c \in \mathbb{N}$ be the maximum height of a rectangle assembled by a system $(T, \tau_i)$.
Let $g_{\rm{max}}$ be the maximum strength of a glue in $T$.
We exhibit an infinite set of integers $n$ for which $f(n) \leq 4g_{\rm max}\sqrt{n}$, implying $f(n)$ cannot be $\omega(\sqrt{n})$.

Let $\tg$ be a terminal supertile of $(T, \tau_i)$, and thus a $k_i \times k_i$ square.
Let $n = \lfloor k_i^2/2 \rfloor$.
Then there exists a cut of $\tg$ into supertiles $\ta$, $\tb$, such that $n = |\ta| \leq |\tb|$ and the cut contains at most $k_i+1$ edges.
Then since $\tg$ is stable, $\tau_i(n) \leq (k_i+1)g_{\rm max} \leq (\sqrt{2n+2} + 1)g_{\rm max}$.
So for each $k_i$, there exists a unique $n$ such that $f(n) \leq \tau_i(n) \leq 4g_{\rm max}\sqrt{n}$.
\end{proof}

%% file: instability.tex
\section{$\tau$-stabilility is \coNP-complete}
\label{sec:tau-instability}

In two-handed tile assembly systems that are not size-dependent, determining whether a supertile is $\tau$-stable amounts to determining if there exists a cut of the bond graph of weight less than $\tau$, a problem decidable in polynomial time.
In contrast, we prove that the same problem is \coNP-complete for size-dependent systems, even when restricted to constant-time-computable temperature functions with just two distinct temperatures.


\begin{theorem}
Given a temperature function $\tau : \mathbb{N} \rightarrow \mathbb{N}$ and supertile, determining whether the supertile is $\tau$-stable is \coNP-complete. 
\end{theorem}

\begin{proof}
The problem is clearly in~\coNP, since verifying a break only requires comparing~$\tau$ of the sizes of the resulting supertiles to the strength of the bonds in the cut.
The \coNP-hardness of the problem is proved via a reduction from maximum independent set in Hamiltonian cubic (3-regular) planar graphs, proved \NP-hard in~\cite{Fleischner-2010}.\footnote{Instances of this problem include a Hamiltonian cycle of the graph.}

The reduction constructs a rectangular supertile with adjacent tiles sharing extremely strong bonds, except for a region of weakness running horizontally and partitioning the supertile into two equal-sized halves. 
This weakness runs through \emph{vertex gadgets} (see Figure~\ref{fig:instability-vertex-gadget}).

\begin{figure}[ht]
\centering
\includegraphics[scale=1.0]{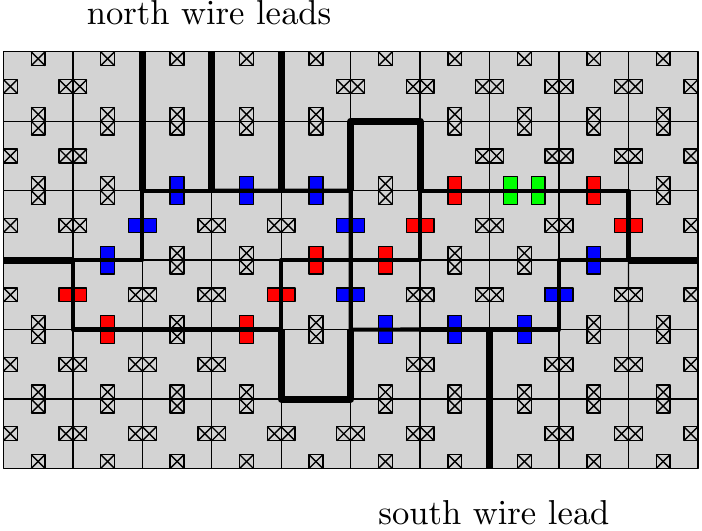}
\caption{The vertex gadget of the $\tau$-stability \coNP-hardness reduction. 
The gray glues have strength too large to be in any cut.}
\label{fig:instability-vertex-gadget}
\end{figure}

\textbf{Vertex gadgets.}
Vertex gadgets have two cut segments through them: \emph{include} and \emph{exclude} (blue and red, respectively, in Fig.~\ref{fig:instability-vertex-gadget}).
Include and exclude cut segments have the same number of bonds and place the same number of tiles on opposite sides of the cut segment.
All bonds on these segments are strength-1, except a special strength-2 bond on the exclude segment (green in Fig.~\ref{fig:instability-vertex-gadget}).
The bond structure of the vertex gadget ensures that any cut of the gadget disconnecting the upper and lower halves that keeps each half connected must be either the include or exclude segment.

Four short 0-strength cut segments called \emph{wire leads} extend vertically from the include segment.
Wire leads of distinct vertex gadgets are connected together by long 0-strength cut segments called \emph{wires}.
The wiring scheme is described next.

\textbf{Layout.}
Because the input graph is planar, Hamiltonian, and cubic, it can be drawn orthogonally in a $|V| \times |V|$ grid such that no two edges cross, all vertices lie on a horizontal line in the order they appear along a Hamiltonian cycle (see Figure~\ref{fig:cubic-drawing}), and all edges not in this cycle lying exclusively above or below the horizontal line containing the vertices.

\begin{figure}[ht]
\centering
\includegraphics[scale=1.0]{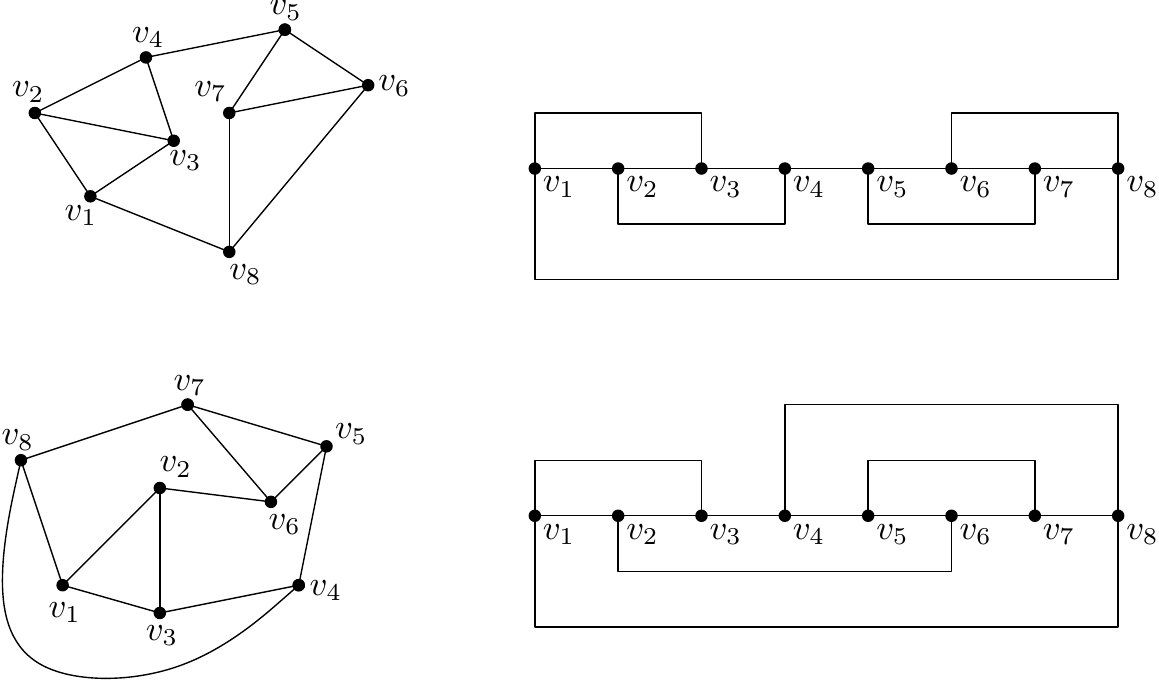}
\caption{A cubic planar Hamiltonian graph (left) and an orthogonal drawing of the graph in a $|V| \times |V|$ grid with collinear vertices and edges either above or below the line containing the vertices.}
\label{fig:cubic-drawing}
\end{figure}

This drawing of the input graph serves as a template for constructing the supertile (see Figure~\ref{fig:instability-schematic}).
The vertices are replaced by a row of vertex gadgets of width $10|V|$, the total width of the supertile, and are connected as in the drawing.
Adjacent gadgets on the Hamiltonian cycle are connected via the left and right north wire leads of the gadget, while non-adjacent gadgets are connected via the middle north or south wire leads.

\begin{figure}[htb]
\centering
\includegraphics[scale=1.0]{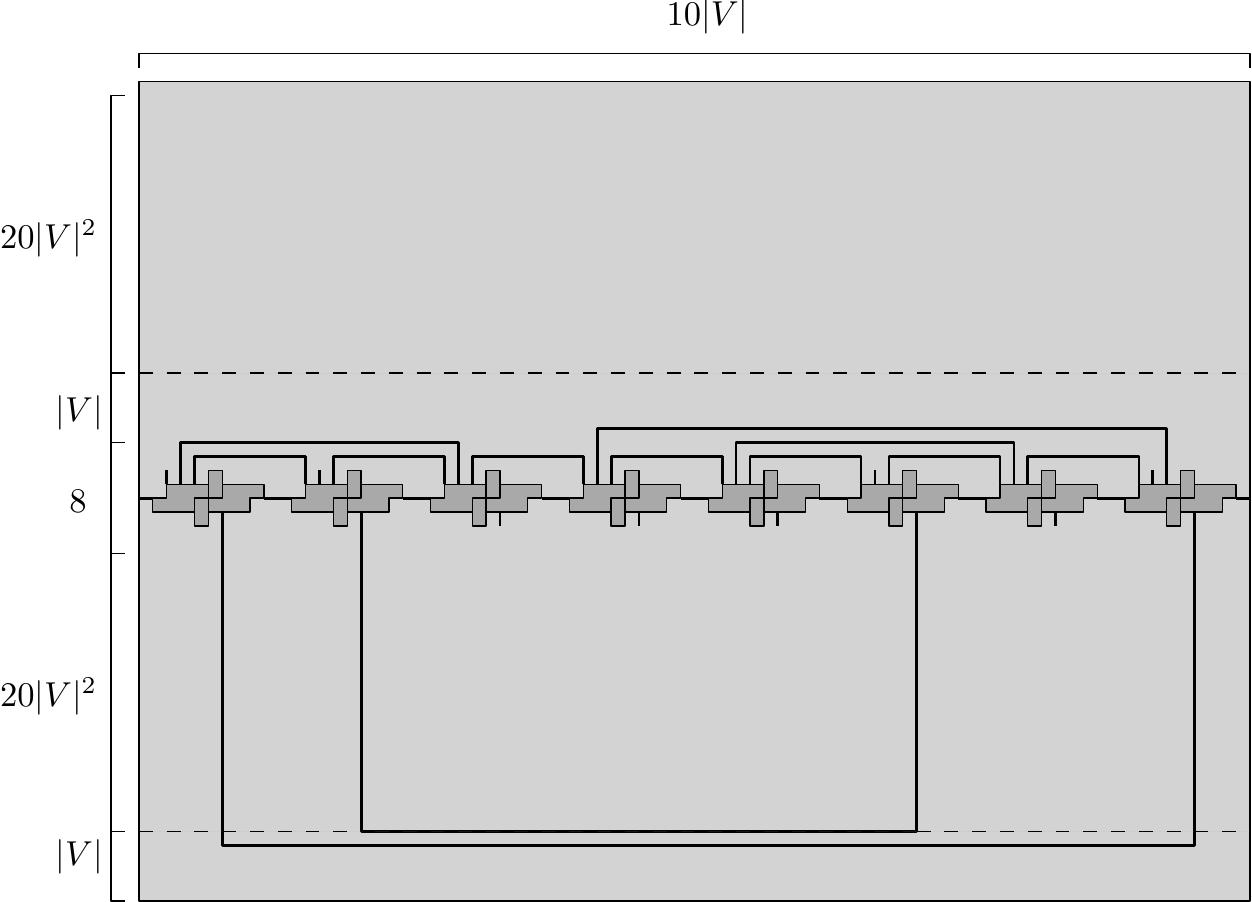}
\caption{A schematic of the supertile constructed by the reduction.
The input graph and drawing of this instance is that seen in the lower half of Figure~\ref{fig:cubic-drawing}.}
\label{fig:instability-schematic}
\end{figure}

Wires connecting north wire leads are contained in the $10|V| \times |V|$ region above the vertex gadgets corresponding to the upper half of the graph drawing.
Wires connecting south wire leads are contained in a $10|V| \times (20|V|^2 + |V|)$ region below the vertex gadgets.
The lowermost $|V|$ rows follow the drawing, while the upper $20|V|^2$ rows of this region only extend the south wire leads vertically.

\textbf{Temperature.}
In total, the supertile has size $s = 10|V| \cdot (40|V|^2 + 2|V| + 8)$.
Given an input instance with integer $k$ (``Does there not exist an independent set of size at least $k$?''), the temperature function used is:
\begin{displaymath}
\tau(n) = \left\{
\begin{array}{ll}
1 & : n < s/2\\
11|V|-k+1 & : {\rm otherwise} 
\end{array}
\right.
\end{displaymath}
Recall that a cut is a break cut if the cut has less than some strength specified by the temperature function.
Since $\tau(n) = 1$ for all $n < s/2$, the supertile may only break into supertiles of exactly equal size.

\textbf{Correctness.}
Consider restrictions on break cuts of the supertile.
The tiles in the northwest and southwest corners of the supertile cannot be in a common resulting supertile of the break, since such a half must have size at least $10|V| \cdot 20|V|^2 + 6 \cdot 20(|V|^2+|V|) > s/2$.
Thus any break cut must disconnect these two tiles.
Call the resulting supertiles of a break containing the tiles in the northwest and southwest corners of the supertile the \emph{upper half} and \emph{lower half}. 
Since each half must be connected, each contains the same number of tiles from each vertex gadget. 

Suppose the lower half contains tiles lying above the vertex gadgets.
Each tile is in a \emph{face} of the graph drawing formed by a maximal set of tiles connected by glues too strong to be in any cut.
Since no such face contains the northwest corner of the supertile, none of these faces are the outer face of the drawing.
So all such tiles lie in a region of size $10|V| \cdot (|V|+8) < 20|V|^2$.
Thus the lower half contains less than $20|V|^2$ tiles lying above the vertex gadgets.

On the other hand, every face below the vertex gadgets has size more than $20|V|^2$.
So the number of tiles contained in the upper half lying below the vertex gadgets is either more than $20|V|^2$ or~0.
Since the two halves have equal size and numbers of tiles in each vertex gadget, they must also have an equal number of tiles lying on the opposite side of the vertex gadgets, namely~0.
So the halves are separated by a path through the vertex gadgets traversing either the include or exclude cut segment in each gadget. 

No two include segments connected by a wire can both be included in a break cut, as the result is a disconnected upper (north wire) or lower (south wire) half.
Moreover, the strength of a cut using $i$ include segments is $10i + 11(|V|-i) = 11|V|-i$.
So there exists a cut of strength less than $11|V|-k+1$ between two connected supertiles both of size $s/2$ if and only if there exists a cut using $i \geq k$ include segments that avoids any two include segments connected by a wire.
Thus the supertile is $\tau$-stable if and only if there does not exist an independent set of size $k$ or larger in the input graph. 
\end{proof}

%% file: openProblems.tex
\section{Open Problems}
\label{sec:openProblems}

The rectangle and square constructions in this work use artificial temperature functions engineered in tandem with the tile sets.
A central open question is whether physically implementable families of temperature functions (e.g. $\tau(n) = cn^{\delta}$ for varying $c, \delta > 0$) are similarly capable of such control.
We conjecture that the design of such systems is possible but difficult; consider the lengthy analysis of the construction in Section~\ref{sec:line} with just~5 components.
Alternatively, temperature functions may be given as input along with shapes, with the goal of designing systems that assemble shapes \emph{despite} the temperature functions.

The difficulty of system design is supported by the \coNP-hardness of determining stability. 
Proving the \PSPACE-hardness of predicting a system outcomes, such as whether a unique terminal supertile exists, would give even further evidence of this difficulty.

As previously discussed, \emph{reversibility} is a key feature of size-dependent systems.
Reversibility has been more directly incorporated into algorithmic design in other tile assembly models, leading to functionality not found in irreversible models.
For instance, \emph{replication} of shapes and patterns~\cite{Abel-2010a,Keenan-2013}, \emph{fuel-efficient} systems~\cite{Padilla-2013a,Schweller-2013a}, and assembly of arbitrary shapes using a small, bounded scale factor~\cite{Demaine-2011b}.
Can any of these be achieved with size-dependent systems?